\documentclass[12pt,reqno]{amsart}
\usepackage[comma]{natbib}
\usepackage[a4paper,left=1in,right=1in,top=1.25in,bottom=1.25in]{geometry}
\usepackage{color}
\usepackage{graphicx}
\usepackage{amsmath}
\usepackage{amssymb}
\usepackage{enumerate}
\usepackage{comment}

\newtheorem{theorem}{Theorem}%[section]

\newtheorem{lemma}[theorem]{Lemma}
\theoremstyle{definition}
\newtheorem*{definition}{Definition}
\theoremstyle{remark}

\begin{document}

\title[]{Re-examining aggregation in the Tallis-Leyton model of parasite acquisition} 

\author[]{R. McVinish}
\address{School of Mathematics and Physics, University of Queensland}
\email{r.mcvinish@uq.edu.au}


\begin{abstract}
The Tallis-Leyton model is a simple model of parasite acquisition where there is no interaction between the host and the acquired parasites. We examine the effect of model parameters on the distribution of the host's parasite burden in the sense of the Lorenz order. This fits with an alternate view of parasite aggregation that has become widely used in empirical studies but is rarely used in the analysis of mathematical models of parasite acquisition.
\end{abstract}

\maketitle

\section{Introduction} % Initial capital letter, then lower case. No full stop.

The distribution of parasites among its host population typically displays a high degree of inequality. In fact, this phenomenon, called aggregation, is almost a universally observed \citep{SD:95, Poulin:07}. Despite its importance, aggregation lacks a universally accepted definition. Instead, the phenomenon is studied using a number of common measures of aggregation that summarise different properties of the parasite's distribution \citep{Pielou:77,ML:2020, MPF:2023}. The most commonly used measures of aggregation in theoretical models are the $k$ parameter of the negative binomial distribution \citep{AM:78a, AM:78b, RP:02, Schreiber:06, McPherson:12} and the variance-to-mean ratio \citep{Isham:95, BP:00, HI:2000, Peacock:18}. Both of these measures can be interpreted as quantifying how far the distribution of parasites deviates from a Poisson distribution.

An alternative view of aggregation was put forward by \citet{Poulin:1993}, arguing that a measure of the discrepancy between the observed distribution of parasites in the hosts and the ideal distributions where all hosts are infected with the same number of parasites would be the best measure of aggregation. This view puts the Lorenz ordering of distributions \citep{Lorenz:1905,Arnold:87} central in the study of aggregation. As a measure of this discrepancy, Poulin proposed his index, {\em D}, which is essentially an estimator of the Gini index \citep{Gini:24}. This has since become  become one of the standard measure of aggregation used in studies of wild parasite populations. Although the Gini index has been commonly used in empirical studies since Poulin's proposal, we are unaware of any examination of its behaviour in theoretical models of parasite acquisition.

The aim of this paper is to apply Poulin's view of aggregation to a simple stochastic model for the number of parasites in a definitive host proposed by \citet{TL:1969}. The Tallis-Leyton model assumes that the host makes infective contacts following a Poisson process and at each infective contact a random number of parasites enter the host. Once a parasite enters the host, it is alive for a random period of time. The lifetimes of parasites, numbers of parasites entering the host at infective contacts, and the Poisson process of infective contacts are all assumed to be independent. Furthermore, the parasites are assumed to have no effect on the host mortality. 

In Section 2 we review some background on the Lorenz ordering and the closely related convex ordering of distributions. The Tallis-Leyton model is analysed in Section 3. We first show show how the host's parasite burden can be represented as a compound Poisson distribution. This representation is then applied to determine how each of the model parameters affect the Lorenz ordering of the distribution of parasites in the host. The final part of the analysis shows that the host's parasite burden is asymptotically normally distributed in the limit as the rate of infectious contacts goes to infinity. The paper concludes with a discussion of future challenges in analysing models of parasite aggregation.

\section{Background}

\subsection{Tallis-Leyton model}

\citet{TL:1969} proposed the following model for the parasite burden $M(a)$ of a definitive host at age $a$, conditional on survival of the host to age $a$. The host is assumed to be parasite free at birth so $M(0) = 0$. The host makes infective contacts during its lifetime following a Poisson process with rate parameter $\lambda$. At each infective contact, the number of parasites that enter the host is a random variable $N$ and once a parasite enters the host it survives for a random period $T$. The parasites have no effect on the host mortality.  The lifetimes of parasites, numbers of parasites entering the host at infective contacts, and the Poisson process of infective contacts are all assumed to be independent. Although we won't make use of this fact, we note that this process also describes an infinite server queue with bulk arrivals and general independent service times \citep{HCK:83}.

We let $G_X$, $F_X$ and $\bar{F}_X = 1 - F_{X}$ denote the probability generating function (PGF), distribution function, and survival function of a random variable $X$. We write $G_{M}(\cdot;a)$ for the PGF of $ M(a)$, conditional on survival of the host to age $ a$. \citet{TL:1969} showed
\[
G_{M}(z;a) = \exp \left(\lambda \int_{0}^{a} \left[G_{N}(1 + \bar{F}_{T}(a-s) (z-1)) -1 \right] ds\right).
\]
Standard arguments show that the mean and variance of $M(a)$ are
\begin{align*}
    \mu(a) & = \lambda \mathbb{E} N \int_{0}^{a} \bar{F}_{T}(s)\, ds \\
    \sigma^{2}(a) & = \lambda \mathbb{E} N(N-1) \int_{0}^{a} \bar{F}_{T}^2(s)\, ds + \lambda \mathbb{E} N \int_{0}^{a} \bar{F}_{T}(s)\, ds.
\end{align*}
Hence, the variance-to-mean ratio is
\begin{equation} \label{eq:TL_VMR}
\text{VMR}(M(a)) = 1 + \frac{\mathbb{E} N(N-1)}{\mathbb{E} N} \frac{\int^a_0 \bar{F}^{2}_{T}(s) ds}{\int^a_0 \bar{F}_{T}(s) ds}.
\end{equation}
Assuming $\mathbb{E}N < \infty$ and $\mathbb{E} T < \infty $, the limiting distribution of parasite burden as $ a\to\infty$ exists and has PGF 
\[
G_{M}(z,\infty) = \exp \left(\lambda \int_{0}^{\infty} \left[G_{N}(1 + \bar{F}_{T}(s) (z-1)) -1 \right] ds\right).
\]
Since scaling of the host age, the lifetime distribution of parasites, and the inverse of the rate of infective contacts by a common factor results in the same distribution for the host's parasite burden, we assume that the expected lifetime of a parasite is 1. 

\subsection{Convex order and Lorenz order}

\citet{Lorenz:1905} proposed the Lorenz curve as a graphical measure of inequality. The following general definition of the Lorenz curve was given by \citet{Gastwirth:1971}.

\begin{definition}
    The Lorenz curve $L: [0,1] \to [0,1] $ for the distribution $F$ with finite mean $ \mu$ is given by
    \[
        L(u) = \frac{\int_{0}^{u} F^{-1}(y)\, dy}{\mu}, 
    \]
    where $ F^{-1}$ is the quantile function
    \[
        F^{-1}(y) = \sup \{ x : F(x) \leq y \} \quad \text{for } y \in (0,1).
    \]
\end{definition}
Adapting the description in  \citet[Section 3.1]{Arnold:87} to a parasitology context, the Lorenz curve $L(u)$ represents the proportion of the parasite population infecting the least infected $u$ proportion of the host population. When all hosts are infected with the same number of parasite, the Lorenz curve is given by $L(u) = u$ and is called the egalitarian line. Note that $L_{X}(u) \geq u$ for all $ u \in [0,1]$. 

The Lorenz curve defines a partial order on the class of all distributions on $[0,\infty)$ with finite mean \citep[Definition 3.2.1]{Arnold:87}.

\begin{definition}
    Let $ X$ and $Y$ be random variables with the respective Lorenz curves denoted $ L_{X} $ and $ L_{Y}$. We say $ X$ is smaller in the Lorenz order, denoted $ X \leq_{\rm{Lorenz}} Y $ if $ L_{X}(u) \geq L_{Y}(u)$ for every $ u \in [0,1]$.
\end{definition} 

The Lorenz curves of some standard distributions are given in \citet[Section 6.1]{Arnold:87}. The negative binomial distribution is extensively used in parasitology. When negative binomial distributions are parameterised in terms of the mean $m$ and $k$, they can be compared in the Lorenz order \citep{ML:2024}. Specifically, let $ \mathsf{NB}(m,k)$ denote the negative binomial distribution with PGF
\[
G(z) = \left(\frac{k}{k + m - mz} \right)^{k}.
\]
Then 
\begin{itemize}
    \item[(i)] for any $ k >0 $ and $ 0< m_{1} < m_{2}$, $\mathsf{NB}(k,m_{2}) \leq_{\rm Lorenz} \mathsf{NB}(k, m_{1})$, and 
    \item[(ii)] for any $ m>0  $ and $ 0< k_{1} < k_{2} $, $\mathsf{NB}(k_{2},m) \leq_{\rm Lorenz} \mathsf{NB}(k_{1},m)$.
\end{itemize}

Closely related to the Lorenz order is the convex order of distributions.

\begin{definition}
    Let $X$ and $Y$ be two random variables such that $\mathbb{E} X = \mathbb{E} Y$. We say $X$ is small than $Y$ in the convex order, denoted $X \leq_{\rm cx} Y$, if $ \mathbb{E} \phi(X) \leq \mathbb{E} \phi(Y)$ for all convex functions $ \phi: \mathbb{R} \to \mathbb{R}$, provided the expectations exist. 
\end{definition}

These two orderings are related since $ X \leq_{\rm Lorenz} Y$ if and only if 
\[
\mathbb{E}\, \phi\left(\frac{X}{\mathbb{E}X}\right) \leq \mathbb{E}\, \phi\left(\frac{Y}{\mathbb{E}Y}\right)
\]
for every continuous convex function $\phi$ \citep[Corollary 3.2.1]{Arnold:87}. In other words,
\begin{align*}
X &\leq_{\rm Lorenz} Y    & \text{is equivalent to}&  &\frac{X}{\mathbb{E} X} &\leq_{\rm cx} \frac{Y}{\mathbb{E} Y}.
\end{align*}

\citet[Section 3.A]{SS:07} provide an extensive review of results on the convex order. We briefly mention some of the important results that are used in our analysis. 
\begin{itemize}
    \item The convex order is {\em closed under weak limits} provided the expectations also converge \citep[Theorem 3.A.12 (c)]{SS:07}.
    \item  The convex order is {\em closed under mixtures} \citep[Theorem 3.A.12 (b)]{SS:07}. Let $X$, $Y$, and $\Theta$ be random variables and write $[X \mid \Theta = \theta] $ and $ [Y \mid \Theta = \theta]$ for the conditional distributions of $X$ and $Y$ given $ \Theta = \theta$. If $ [X \mid \Theta = \theta] \leq_{\rm cx} [Y \mid \Theta = \theta]$ for all $\theta$ in the support of $\Theta$, then $X \leq_{\rm cx} Y$. As an application of this property we can say that if $X \leq_{\rm cx} Y$ and $Z$ is an independent non-negative random variable, then $ ZX \leq_{\rm cx} ZY$.
    \item The convex order is {\em closed under convolutions} \citep[Theorem 3.A.12 (d)]{SS:07}. Let $X_1, X_2,\ldots, X_k$  and $Y_1, Y_2, \ldots, Y_k$ be two sets of independent random variables . If $X_j \leq_{\rm cx} Y_{j}$ for $j=1,2,\ldots,k$, then 
    \[
    \sum_{j=1}^{k} X_{j} \leq_{\rm cx} \sum_{j=1}^{k} Y_{j}.
    \]
    \item Combining the properties of closure under mixtures and closure under convolutions, we see the convex order is {\em closed under random sums} so 
    \[
    \sum_{j=1}^{K} X_{j} \leq_{\rm cx} \sum_{j=1}^{K} Y_{j},
    \]
    for any non-negative integer random variable $K$. As an application of the closure under random sums property of the convex order, consider two random variables $K$ and $\tilde{K}$ that related by binomial thinning. That is, $G_{\tilde{K}}(z) = G_{K}(1 - p + pz)$ for some $p \in (0,1)$. Then   $ K \leq_{\rm Lorenz} \tilde{K}$ \citep[Section 3]{ML:2020}
    \item The closure under random sums property can be adapted to the case where the $X_1, X_2, \ldots$ and $Y_1, Y_2, \ldots$ are two iid sequences with $X \leq_{\rm cx} Y$, and $K_1$ and $K_2$ are non-negative integer random variables such that $K_{1} \leq_{\rm cx} K_{2}$. In this case, \citet[Theorem 3.A.13]{SS:07} implies 
    \[
    \sum_{j=1}^{K_1} X_{j} \leq_{\rm cx} \sum_{j=1}^{K_2} Y_{j}.
    \]
    \item The survival function can be used to establish if two random variables can be compared in the convex order. If $X$ and $Y$ are two random variables with the same mean, then $X \leq_{\rm cx} Y$ if $\bar{F}_X - \bar{F}_Y$ has a single sign change and the sign sequence is $ +, -$ \citep[Theorem 3.A.44(b)]{SS:07}. This property can also be used to characterise the convex order \citep[Theorem 3.A.45]{SS:07}.
\end{itemize}

\subsection{Measures of aggregation}
In practice, levels of aggregation are compared with numerical summaries rather than using the entire Lorenz curve. A useful measure of aggregation respects the Lorenz ordering. That is, if $I(\cdot)$ is a measure of aggregation that respects the Lorenz ordering and $X \leq_{\rm Lorenz} Y$, then $I(X) \leq I(Y)$. \citet[Chapter 5]{Arnold:87} review several inequality measures and these can be applied as measures of aggregation. We restrict our attention in this paper to the Gini index, the Hoover index (also known as the Pietra index, or the Robin-Hood index),  $ 1 - \text{prevalence} $, and the coefficient of variation.

The Gini index \citep{Gini:24} is given by twice the area between the egalitarian line and the Lorenz curve. For a random variable $X$, the Gini index can be expressed as
\[
G(X) = \frac{\mathbb{E}\, |X - \tilde{X}|}{2\, \mathbb{E} X},
\]
where $ \tilde{X}$ is an independent random variable with $ \tilde{X} \stackrel{d}{=} X $. The Hoover index is given by the maximum vertical distance between the egalitarian line and the Lorenz curve. \citet{ML:2020} argue that this index could be useful due to its simple interpretation as the proportion of the parasite population that would need to be redistributed among the hosts in order for all hosts to have the same parasite burden. The Hoover index can be expressed as
\[
H(X) = \frac{\mathbb{E}\left|X - \mathbb{E} X\right|}{2\, \mathbb{E} X}
\]
Prevalence, the probability that a host is infected by at least one parasite, is an important quantity in parasitology. Although prevalence is not usually thought of as a measure of aggregation, we may express $ 1 - \text{prevalence} $ in terms of the  Lorenz curve as
\[
1 - \text{prevalence} = \max\{u: L(u) = 0 \}.
\]
For the Tallis-Leyton model,
\begin{equation}
1 - \text{prevalence} = G_{M}(0;a) = \exp \left(\lambda \int_{0}^{a} \left[G_{N}(1 - \bar{F}_{T}(a-s)) -1 \right] ds\right). \label{eq:1minusPrev}
\end{equation}
Finally, the coefficient of variation is given by
\[
CV(X) = \frac{\sqrt{\text{Var}(X)}}{\mathbb{E} X}.
\]
This measure is rarely used in parasitology, though it is mentioned in some reviews on parasite aggregation such as \citet{Wilson:2001} and \citet{ML:2020}. As means and variances are commonly reported in empirical studies and are often easily calculated for theoretical models, it may be useful in some contexts. For example, the squared coefficient of variation for the Tallis-Leyton model is 
\begin{equation} \label{eq:TL_CV}
\text{CV}^{2}(M(a)) = \frac{1}{\lambda} \frac{\mathbb{E} N(N-1)}{(\mathbb{E} N)^{2}} \frac{\int^a_0 \bar{F}^{2}_{T}(s) ds}{ \left( \int^a_0 \bar{F}_{T}(s) ds \right)^{2}} + \frac{1}{\lambda \mathbb{E} N \int^a_0 \bar{F}_{T}(s) ds}.   
\end{equation}
These indices are related by the following inequality
\[
1 - \text{prevalence} \leq H(X) \leq G(X) \leq H(X) (2- H(X)) \leq CV(X)
\]
\citep{Taguchi:68, ML:2020}. In particular, if $ \mathbb{E}X \leq 1$, then we have the equality $1 - \text{prevalence} = H(X)$. The Gini index and Hoover index can be further related to the coefficient of variation when the distribution of parasites is approximately normal. Suppose $X_{1}, X_{2}, \ldots $ is a sequence of random variables such that 
\[
\frac{X_{n} - \mathbb{E}X_{n}}{\sqrt{\text{Var}(X_{n})}} =: Z_{n} \stackrel{d}{\to} Z, 
\]
where $Z \sim \mathsf{N}(0,1)$. As $X_{n} \geq 0 $ with probability one, the above limit is only possible if  $CV(X_n) \to 0$. Nevertheless, the ratio of the Hoover index to the coefficient of variation still has a well defined limit. The Hoover index of $X_n$ can be expressed as
\[
H(X_{n}) = \frac{\mathbb{E} \left|X_{n} - \mathbb{E} X_{n}\right|}{2 \mathbb{E} X_{n}}  = \frac{\sqrt{\text{Var}(X_{n})}}{2\mathbb{E} X_{n}} \mathbb{E}|Z_{n}|.
\]
Since $ \mathbb{E} Z_{n}^{2} = 1$, this collection of random variables is uniformly integrable and $ \mathbb{E} |Z_{n}| \to \mathbb{E}|Z| = \sqrt{2/\pi} $. Hence,
\begin{equation}
\frac{H(X_{n})}{CV(X_{n})} \to \frac{1}{\sqrt{2\pi}}.  \label{eq:Hoover_approx}
\end{equation}
Similarly, the Gini index of $X_n$ can be expressed as
\[
G(X_{n}) = \frac{\mathbb{E} |X_{n} - \tilde{X}_{n}|}{2 \mathbb{E} X_{n}}  = \frac{\sqrt{\text{Var}(X_{n})}}{2\mathbb{E} X_{n}} \mathbb{E}|Z_{n} - \tilde{Z}_{n}|,
\]
where $\tilde{Z}_{n}$ is an independent random variable with $\tilde{Z}_{n} \stackrel{d}{=} Z_{n}$. Applying the asymptotic normality and uniform integrability of the $Z_{n}$, 
\begin{equation}
\frac{G(X_{n})}{CV(X_{n})} \to \frac{1}{\sqrt{\pi}} .  \label{eq:Gini_approx}
\end{equation}

\begin{comment}
{\bf Example:} Suppose $X_n \sim \mathsf{Poi}(\lambda_n)$, with $ \lambda_n \to \infty$. Convergence of the Poisson distribution to the normal distribution is well known. The coefficient of variation is $CV(X_n) = \lambda^{-1/2}_n$. Expressions for the Hoover and Gini indices
\[
H(X_n) = \frac{\lceil \lambda_n \rceil}{\lambda_n} e^{-\lambda_n} \frac{\lambda_n^{\lceil \lambda_n \rceil}}{\lceil \lambda_n \rceil!} \quad \quad G(X_n) = {}_1F_1(\tfrac{1}{2}; 2; -4\lambda_n).
\]
The limiting behaviour follows from Stirling's approximation.
\end{comment}
%\[
%\frac{H(X_n)}{CV(X_n)} =  \frac{\lceil \lambda_n \rceil}{\sqrt{\lambda_n}} e^{-\lambda_n} \frac{\lambda_n^{\lceil \lambda_n \rceil}}{\lceil \lambda_n \rceil!} \sim  \frac{\lceil \lambda_n \rceil}{\sqrt{\lambda_n}} e^{-\lambda_n} \frac{\lambda_n^{\lceil \lambda_n \rceil}}{\sqrt{2 \pi \lceil \lambda_n \rceil} (\lceil \lambda_n \rceil/ e)^{\lceil \lambda_n \rceil}} =  \frac{\lceil \lambda_n \rceil}{\lambda_n} e^{-\lambda_n} \frac{\lambda_n^{\lceil \lambda_n \rceil}}{\sqrt{2 \pi \lceil \lambda_n \rceil} (\lceil \lambda_n \rceil/ e)^{\lceil \lambda_n \rceil}} 
%\]

\subsection{Numerical evaluation}
From equation (\ref{eq:TL_CV}), the coefficient of variation can be relatively easily evaluated for the Tallis-Leyton model. Numerical integration of $ \bar{F}_{T}(s)$ and $\bar{F}_{T}^2(s)$ may be required, but the dependence on age and $\lambda $ is explicit. Similarly, $ 1 - \text{prevalence} $ could be evaluated with a single numerical integration using (\ref{eq:1minusPrev}). On the other hand, evaluation of the Hoover and Gini indices require evaluation of the probability mass function. In the examples of the next section, we numerically evaluate the probability mass function of $M(a)$ by inverting $G_{M}(z;a)$ using the Abate-Whitte algorithm \citep{AB:92}. The algorithm was implemented in MATLAB \citep{MATLAB} using the \texttt{vpa} function in the Symbolic Math Toolbox \citep{vpa} for high precision arithmetic.

\section{Analysis of the Tallis-Leyton model}

The analysis begins with a representation of the host's parasite burden as a compound Poisson distribution. This representation is used extensively to understand how the rate of infective contacts, the distribution of the number of parasites that enter the host during an infective contact, the age of the host, and lifetime distribution of the parasites all affect the distribution of parasites in the host in terms of the Lorenz order. When comparing the host's parasite burden in two systems, the parameters of the second parasite-host system is distinguished by a tilde.

\subsection{Compound Poisson representation}

Let $U_{1}, U_{2},\ldots $ be a sequence of independent standard uniform random variables and define $X:\mathbb{Z}_{\geq 0} \times [0,1] \to \mathbb{Z}_{\geq 0}$ by
\[
X(n,v) := \sum_{k=1}^{n} \mathbb{I}(U_{k} \leq v),
\]
with $X(n,v) = 0$ when $ n=0$. For given $n$ and $v$ the distribution of $X(n,v)$ is  $\mathsf{Bin}(n,v)$.  
%Let $\Lambda(t)$ be a Poisson process with rate $\lambda$ so $\Lambda(t) \sim \mathsf{Poi}(\lambda t)$.

\begin{theorem} \label{thm:Rep}
Assume $ T$ has a continuous distribution. Define $ V$ to be a random variable on $[\bar{F}_{T}(a),1] $ with distribution function
\begin{equation}
F_{V}(v) = 1 - a^{-1} \bar{F}^{-1}_{T}(v) \label{eq:V}
\end{equation}
Let $X_{1},X_{2},\ldots$ be a sequence of independent random variables with the same distribution as $X(N,V)$, where $N$ and $V$ are independent.  Let $\Lambda(t)$ be a Poisson process with rate $\lambda$ so $\Lambda(t) \sim \mathsf{Poi}(\lambda t)$. Then
\[
M(a) \stackrel{d}{=} \sum_{k=1}^{\Lambda(a)} X_{k}.
\]    
\end{theorem}

\begin{proof}
Standard conditioning arguments show that the PGF of $ X(N,v) $  is 
\[
G_{X(N,v)}(z) = G_{N}\left(1 +  v (z-1) \right).
\]
Hence,
\begin{align*}
G_{X(N,V)}(z)& = \int_{\bar{F}_{T}(a)}^{1} G_{N}\left(1 +  v(z-1) \right) d \left[1 - a^{-1} \bar{F}^{-1}_{T}(v)\right]  \label{Rep:Eq2}  
\end{align*}
Again applying standard conditioning arguments, we see the PGF of $\sum_{k=1}^{\Lambda(a)} X_{k}$ is
\begin{align*}
    G_{ \sum_{k=1}^{\Lambda(a)} X_{k}}(z) & = \exp \left\{ \lambda a \left( \int_{\bar{F}_{T}(a)}^{1} G_{N}\left(1 +  v(z-1) \right) d \left[1 - a^{-1} \bar{F}^{-1}_{T}(v)\right]  -1 \right)   \right\} \\
    & = \exp \left\{ \lambda  \int_{\bar{F}_{T}(a)}^{1} \left[ G_{N}\left(1 +  v(z-1) \right) - 1 \right] d \left[a - \bar{F}^{-1}_{T}(v)\right]    \right\}.
\end{align*}
Upon making the substitution $ v = \bar{F}_{T}(a-s)$, the PGF of $\sum_{k=1}^{\Lambda(a)} X_{k}$ can be expressed as 
\begin{align*}
G_{\sum_{k=1}^{\Lambda(a)} X_{k}}(z)& = \exp \left(\lambda \int_{0}^{a} \left[G_{N}(1 + \bar{F}_{T}(a-s) (z-1)) -1 \right] ds\right) = G_{M}(a;z).
%& = \exp \left(\lambda \int_{\bar{F}_{T}(a)}^{1} \left[G_{N}(1 + u (z-1)) -1 \right] d \left[ a - \bar{F}^{-1}_{T}(u)\right] \right),
\end{align*}
%noting that $\bar{F}_{T}(0) = 1$ as $ T $ is continuous. 
\end{proof}

\subsection{Rate of infective contacts}

Our first comparison result concerns the effect of the rate of infective contacts on the distribution of the host's parasite burden. The rate of infective contacts has no effect on the variance-to-mean ratio (\ref{eq:TL_VMR}), whereas the coefficient of variation is strictly decreasing as the rate of infective contacts increases (\ref{eq:TL_CV}). The following result shows that any index respecting the Lorenz order is decreasing as a function of the rate of infective contacts.

\begin{theorem} \label{thm:compRate}
    If $ \tilde{\lambda} < \lambda $ and all other model parameters are equal, then $ M(a) \leq_{\rm Lorenz} \tilde{M}(a)$.
\end{theorem}

\begin{proof}
    Set $\kappa = \tilde{\lambda}/\lambda$. Let $ X_{1}, X_{2}, \ldots$ be a sequence of independent random variables having the same distribution as $X(N,V)$ and let $B_{1},B_{2},\ldots$ be a sequence of independent $\mathsf{Ber}(\kappa)$ random variables that are also independent of the $X_{k}$. As $ \kappa \leq_{\rm cx} B_{k} $ and the convex order is closed under mixtures, $ \kappa X_{k} \leq_{\rm cx} B_{k} X_{k}$. The PGF of $ B_{k} X_{k}$ is $G_{ B_{k} X_{k}}(z) = \kappa\, G_{X(N,V)}(z) + 1-\kappa$. Let $\Lambda(t)$ be a Poisson process with rate $\lambda$. As the convex order is closed under random sums,
    \[
    \kappa \sum_{k=1}^{\Lambda(a)} X_{k} \leq_{\rm cx} \sum_{k=1}^{\Lambda(a)}  B_{k} X_{k}.
    \]
    By Theorem \ref{thm:Rep}, $ \sum_{k=1}^{\Lambda(a)} X_{k} \stackrel{d}{=} M(a)$. To determine the distribution of $ \sum_{k=1}^{\Lambda(a)}  B_{k} X_{k}$, we evaluate its PGF 
    \begin{align*}
    G_{\sum_{k=1}^{\Lambda(a)}  B_{k} X_{k}}(z) & = \exp\left\{\lambda a\left[(\kappa\, G_{X(N,V)}(z) + 1-\kappa) -1 \right] \right\} \\
    & = \exp\left\{\tilde{\lambda} a \left[ G_{X(N,V)}(z) -1 \right] \right\} = G_{\tilde{M}}(a;z).
    \end{align*}
    Hence, $ M(a) \leq_{\rm Lorenz} \tilde{M}(a)$.     
\end{proof}

Figure \ref{fig:Infection} shows the four indices (Gini, Hoover, $ 1 - \text{prevalence} $, and coefficient of variation) for a host aged 3 where  $ \lambda \in [0.25, 128]$, $ N \sim \mathsf{NB}(1,1)$, and $ T \sim \mathsf{Exp}(1)$. Since the coefficient of variation  (\ref{eq:TL_CV}) is proportional to $\lambda^{-1/2}$, it is not displayed for small values of $\lambda$. We see that all four indices are strictly decreasing as a function of $\lambda$. For $\lambda \leq 1$, $\mu(3) \leq 1$ so the Hoover index and $ 1 - \text{prevalence} $. The Hoover index appears to display some discontinuity in the first derivative at points where the expectation is integer valued. This behaviour is less apparent at larger values of $\lambda$. 

\begin{figure}[h] 
    \centering
    \includegraphics{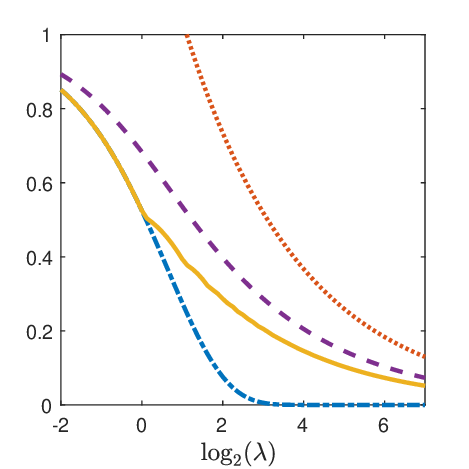}
    \caption{Plot of Hoover index (yellow solid line), Gini index (purple dashed line), $ 1 - \text{prevalence} $ (blue dot-dashed line), and coefficient of variation (orange dotted line) for a host aged 3 in the Tallis-Leyton model with  $ N \sim \mathsf{NB}(1,1)$, and $ T \sim \mathsf{Exp}(1)$.}
    \label{fig:Infection}
\end{figure}

\subsection{Distribution of $N$}

We now consider the role of the distribution of the number of parasites that enter the host during an infective contact. As a concrete example, suppose $N \sim \mathsf{NB}(m,k)$. Then the variance-to-mean ratio is
\[
\text{VMR}(M(a)) = 1 + c m (1+1/k),
\]
where $c>0$ is a constant depending on $F_T$. From this expression we see that the variance-to-mean ratio is increasing in $m$ but decreasing in $k$. In contrast, the coefficient of variation of $M(a)$ is decreasing in both $m$ and $k$. The next two results show that the distribution of the host's parasite burden is decreasing in the Lorenz order as functions of both $m$ and $k$. The first of these results requires the distributions being compared to have the same expectation.

\begin{theorem} \label{thm:compN1}
    Suppose $ N \leq_{\rm cx} \tilde{N}$ and all other model parameters are equal. Then $ M(a) \leq_{\rm cx} \tilde{M}(a)$.
\end{theorem}

\begin{proof} 
    Using an extension of the closure under random sums property of the convex order \citet[Theorem 3.A.13]{SS:07}, 
    \[
    X(N,v) \leq_{\rm cx} X(\tilde{N},v).    
    \]
    As the convex order is closed under mixtures, $ X(N,V) \leq_{\rm cx} X(\tilde{N},V)$. 
    Let $ X_{1}, X_{2}, \ldots$ be a sequence of independent random variables having the same distribution as $X(N,V)$ and let $ \tilde{X}_{1}, \tilde{X}_{2},\ldots $ be a sequence of independent random variables having the same distribution as $X(\tilde{N},V)$. As the convex order is closed under random sums, 
    \[
    \sum_{k=1}^{\Lambda(a)} X_{k} \leq_{\rm cx} \sum_{k=1}^{\Lambda(a)} \tilde{X}_{k}. 
    \]
    Theorem \ref{thm:Rep} shows $ M(a) \leq_{\rm cx} \tilde{M}(a)$.
\end{proof}

For distributions with different means, we consider only the case where $N$ and $\tilde{N}$ are related by binomial thinning. Recall that if $G_{\tilde{N}}(z) = G_{N}(1-p+pz)$ for some $p \in (0,1)$, then $ \tilde{N} \leq_{\rm Lorenz} N$.

\begin{theorem} \label{thm:compN2}
    Suppose that $G_{\tilde{N}}(z) = G_{N}(1 - p + pz)$ for some $ p \in (0,1)$ and all other model parameters are equal. Then $ M(a) \leq_{\rm Lorenz} \tilde{M}(a)$.  
\end{theorem}

\begin{proof}
    Let $ U_1, U_2, \dots $ and $U'_1, U'_2, \ldots$ be independent standard uniform random variables. Then standard conditioning arguments show
    \begin{align*}
        X(\tilde{N},v) \stackrel{d}{=} \sum_{j=1}^{N} \mathbb{I}\left(U_{j} \leq v \right) \mathbb{I}\left(U'_{j} \leq p \right).
    \end{align*}
    As the convex order is closed under mixtures, 
    \[
    p\, \mathbb{I}\left(U_{j} \leq v \right) \leq_{\rm cx} \mathbb{I}\left(U_{j} \leq v \right) \mathbb{I}\left(U'_{j} \leq p \right).
    \]
    As the convex order is closed under random sums, $ p X(N,v) \leq_{\rm cx} X(\tilde{N},v)$. Following the same arguments as in the proof of Theorem \ref{thm:compN1}, we see $ p M(a) \leq_{\rm cx} \tilde{M}(a)$. Hence, $ M(a) \leq_{\rm Lorenz} \tilde{M}(a)$.    
\end{proof}

As noted previously, $ \mathsf{NB}(k,m) \leq_{\rm Lorenz} \mathsf{NB}(\tilde{k},\tilde{m})$ if $ \tilde{k} \leq k $ and $ \tilde{m} \leq m$. If $N \sim \mathsf{NB}(k,m)$, $\tilde{N} \sim \mathsf{NB}(\tilde{k},\tilde{m})$ and all other model parameters are equal, then Theorems \ref{thm:compN1} and \ref{thm:compN2} together imply that $M(a) \leq_{\rm Lorenz} \tilde{M}(a)$. Figure  \ref{Fig:NB} shows the Hoover and Gini indices as functions of the negative binomial parameters $m$ and $k$ for a parasite host system with host aged 10, rate of infective contacts 5, and parasite lifetimes following an exponential distribution with mean 1. Both indices are decreasing in both $m$ and $k$ as we expect from the above results. The contours of the  Hoover index display some discontinuity in the first derivative for $m = 1/5$ ($\log_2(m) \approx -2.3$), which corresponds to a mean of 1 for the host. The contours for both the Hoover and Gini indices tend to become parallel to the respective axes as $m \to \infty$ and $k \to \infty$. This is a consequence of the limiting behaviour of the negative binomial distribution \citep{AdlC:94}.

\begin{figure}[h] 
    \centering
    \includegraphics{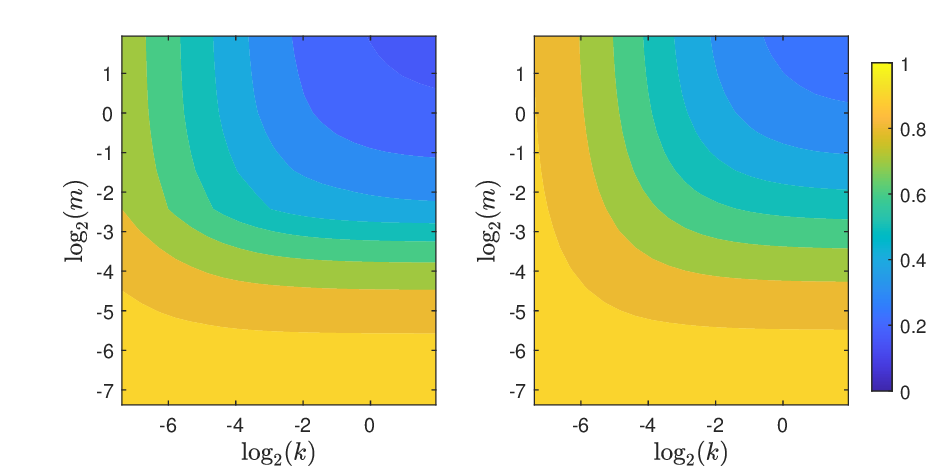}
    \caption{Contour plots showing Hoover index (Left) and Gini index (Right) for a host aged 10 in the Tallis-Leyton model with $\lambda = 5$, $T \sim \mathsf{Exp}(1)$ and $ N \sim \mathsf{NB}(m,k)$.}
    \label{Fig:NB}
\end{figure}

It is natural to consider which distribution for $N$ results in the least aggregated distribution for the host's parasite burden. This requires determining the smallest distribution in the convex ordering. The convex order requires that the distributions compared have the same expected value so let $ n = \mathbb{E} N$. Define the random variable $ N'$ such that
\begin{align*}
\mathbb{P}(N' = \lfloor n \rfloor) & = \lfloor n \rfloor + 1 - n     & \mathbb{P}(N' = \lfloor n \rfloor +1) & = n - \lfloor n \rfloor.
\end{align*}
In the supplementary material of \citet{ML:2020} it was shown that $N' \leq_{\rm cx} N$ so $N'$ is smallest distribution in convex order with expectation $n$. When $n \leq 1$, the smallest distribution in convex order for $N$ leads to $M(a)$ having a Poisson distribution. There is no largest distribution in the convex order.

\subsection{Host age}

Differentiating (\ref{eq:TL_VMR}) with respect to age shows the variance-to-mean ratio is a decreasing function
of age. Since the expected parasite burden is increasing in age, the coefficient of variation is also decreasing in age. The following result shows the host's parasite burden is decreasing in the Lorenz order as a function of age.

\begin{theorem} \label{thm:compAge}
    If $ \tilde{a} <a $, then $  M(a)  \leq_{\rm Lorenz} M(\tilde{a})$.
\end{theorem}

The proof is built from the following lemmas. 

\begin{lemma} \label{lem:bv}
Let $V$ have the distribution (\ref{eq:V}) and let $\tilde{V}$ have the distribution (\ref{eq:V}) with $a$ replaced by $\tilde{a}$. Let $ B \sim \mathsf{Ber}(\mu(\tilde{a})/\mu(a))$ independent of $V$, and let $\tilde{B} \sim \mathsf{Ber}(\tilde{a}/a)$ independent of $\tilde{V}$. Then $BV \leq_{\rm cx} \tilde{B} \tilde{V}$.
\end{lemma}

\begin{proof}
Note that 
\[
\mathbb{E} V = \frac{1}{a} \int_{0}^{a} \bar{F}_{T}(s) ds
\]
so $ \mathbb{E} BV = \mathbb{E} \tilde{B}\tilde{V} $. We show that $BV \leq_{\rm cx} \tilde{B} \tilde{V}$ by examining the sign changes of $\bar{F}_{BV} - \bar{F}_{\tilde{B}\tilde{V}}$. The survival functions of $ BV$ and $\tilde{B}\tilde{V}$ are
\begin{align*}
    \bar{F}_{BV}(w) & = \left\{ \begin{array}{ll}
    \frac{\mu(\tilde{a})}{\mu(a)}, & w \in [0,\bar{F}_{T}(a)) \\
    \frac{\mu(\tilde{a})}{a\, \mu(a)} \bar{F}^{-1}_{T}(w), & w \in [\bar{F}_{T}(a),1) \\
    0, & w > 1,
    \end{array} \right. 
\end{align*}
and 
\begin{align*}
    \bar{F}_{\tilde{B}\tilde{V}}(w) & = \left\{ \begin{array}{ll}
    \frac{\tilde{a}}{a}, & w \in [0,\bar{F}_{T}(\tilde{a})) \\
    \frac{1}{a} \bar{F}^{-1}_{T}(w), & w \in [\bar{F}_{T}(\tilde{a}),1) \\
    0, & w > 1.
    \end{array} \right.
\end{align*}
 Since $\mu(a)$ is increasing in $a$ and $\mu(a)/a$ is decreasing in $a$, 
\[
\frac{\tilde{a}}{a} < \frac{\mu(\tilde{a})}{\mu(a)} <  1.
\]
Hence, $\bar{F}_{BV}(w) - \bar{F}_{\tilde{B}\tilde{V}}(w) > 0$ for all $ w \in [0, \bar{F}_{T}(a)]$. On $[\bar{F}_{T}(a), \bar{F}_{T}(\tilde{a})]$, $ \bar{F}_{\tilde{B}\tilde{V}}(w) = \tilde{a}/a$ whereas $\bar{F}_{BV}$ decreases from $\mu(\tilde{a})/\mu(a) $ to $ \tilde{a} \mu(\tilde{a}) /a \mu(a) < \tilde{a}/a$. For all $ w \geq \bar{F}_{T}(\tilde{a})$,  $\bar{F}_{BV}(w) - \bar{F}_{\tilde{B}\tilde{V}}(w) < 0$.  Hence, $ \bar{F}_{BV} - \bar{F}_{\tilde{B}\tilde{V}}$ has a single sign change and the sign sequence is $ +, -$. Hence,   $BV \leq_{\rm cx} \tilde{B}\tilde{V}$ \cite[Theorem 3.A.44]{SS:07}.
\end{proof}

\begin{lemma} \label{lem:phi}
    For any convex function $\phi$ and any non-negative integer valued random variable $N$ that is independent of $U_{1},U_{2},\ldots$, $\mathbb{E}\phi(X(N,v))$ is a convex function in $v$
\end{lemma}

\begin{proof}
    As the binomial distribution $ \mathsf{Bin}(n,v)$ is a regular exponential family of distribution with expectation linear in $v$, \citet[Proposition 2]{Schweder:1982} implies
    $ \mathbb{E} \phi(X(n,v) $ is convex in $v$ for any positive integer $ n$. As non-negative weighted sums of convex functions are also convex, it follows that $\mathbb{E}\phi(X(N,v))$ is a convex function in $v$.
\end{proof}

\begin{proof}[Proof of Theorem \ref{thm:compAge}]
By construction of $X(n,v)$, if $b$ takes values in $\{0,1\}$, then $b X(n,v) = X(n,bv)$. Applying \citet[Theorem 3.A.21]{SS:07} with Lemmas \ref{lem:bv} and \ref{lem:phi}, 
\[
BX(N,V) = X(N,BV) \leq_{\rm cx} X(N, \tilde{B} \tilde{V}) = \tilde{B} X(N, \tilde{V}).
\]
Since the convex order is transitive and closed under mixtures,
\[
  \frac{\mu(\tilde{a})}{\mu(a)}  X(N,V) \leq_{\rm cx} B X(N,V) \leq_{\rm cx} \tilde{B} X(N, \tilde{V}).
\]
In the notation of Theorem \ref{thm:Rep}, $M(a) \stackrel{d}{=} \sum_{k=1}^{\Lambda(a)} X_{k}$, where $X_1,X_2,\ldots$ is a sequence of independent random variables with $X_k \stackrel{d}{=} X(N,V)$.  From the thinning property of the Poisson process and Theorem \ref{thm:Rep}, we can write $ M(\tilde{a}) \stackrel{d}{=} \sum_{k=1}^{\Lambda(a)} \tilde{B}_{k} \tilde{X}_{k}$, where $\tilde{X}_{1},\tilde{X}_{2},\ldots$ is a sequence of independent random variables with $\tilde{X}_{k}\stackrel{d}{=} X(N,\tilde{V})$ and $\tilde{B}_{1},\tilde{B}_{2},\ldots$ is a sequence of independent $\mathsf{Ber}(\tilde{a}/a)$ random variables that are also independent of the $\tilde{X}_{k}$. As the convex order is closed under random sums, 
\[
\frac{\mu(\tilde{a})}{\mu(a)} M(a) \stackrel{d}{=} \frac{\mu(\tilde{a})}{\mu(a)} \sum_{k=1}^{\Lambda(a)} X_{k} \leq_{\rm cx} \sum_{k=1}^{\Lambda(a)} B_{k} X_{k} \leq_{\rm cx} \sum_{k=1}^{\Lambda(a)} \tilde{B}_{k} \tilde{X}_{k} \stackrel{d}{=} M(\tilde{a})
\]
\end{proof}

\begin{figure}[h] 
    \centering
    \includegraphics{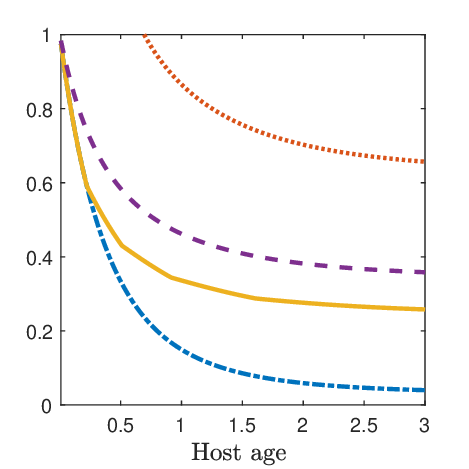}
    \caption{Plot of Hoover index (yellow solid line), Gini index (purple dashed line), $ 1 - \text{prevalence} $ (blue dot-dashed line), and coefficient of variation (orange dotted line) for a host in the Tallis-Leyton model with $\lambda = 5$, $ N \sim \mathsf{NB}(1,1)$, and $ T \sim \mathsf{Exp}(1)$.}
    \label{fig:Age}
\end{figure}

Figure \ref{fig:Age} shows the four indices (Gini, Hoover, $ 1 - \text{prevalence} $, and coefficient of variation) for a host in the Tallis-Leyton model with $\lambda = 5$, $ N \sim \mathsf{NB}(1,1)$, and $ T \sim \mathsf{Exp}(1)$. All four indices are strictly decreasing in host age. As in Figure \ref{fig:Infection}, the Hoover index coincides with $ 1 - \text{prevalence} $ for $\mu(a) \leq 1$ and displays some discontinuity in the first derivative at points where the expectation is integer valued.

\subsection{Parasite lifetime distribution}

Our final comparison result concerns the parasite lifetime distribution. The result shows that variability in the parasite lifetimes decreases the variability in the host's parasite burden. In particular, the result implies that the host's parasite burden is most aggregated when parasites have constant lifetimes.

\begin{theorem}
    Suppose $ \mathbb{E} T = \mathbb{E} \tilde{T}$ and $\bar{F}_{T} - \bar{F}_{\tilde{T}}$ has a single sign change with sign sequence $+, -$ so $T \leq_{\rm cx} \tilde{T}$. Assume all other model parameters are equal, then $\tilde{M}(\infty) \leq_{\rm cx} M(\infty)$.
\end{theorem}

\begin{proof} Assume  For any $a > 0 $ set $\tilde{a}$ such that 
\[
\int_{0}^{a} \bar{F}_{T}(t) dt = \int_{0}^{\tilde{a}} \bar{F}_{\tilde{T}}(t) dt. 
\]
As $\mathbb{E} T = \mathbb{E} \tilde{T}$, it follows that $\tilde{a} > a$. Let $B \sim \mathsf{Ber}(a/\tilde{a})$. Let $V$ have distribution (\ref{eq:V}) and let $\tilde{V}$ have the distribution (\ref{eq:V}) with $a$ replaced by $\tilde{a}$ and $T$ replaced by $\tilde{T}$. The survival functions of $ \tilde{V} $ and $BV$ are
\begin{align*}
    \bar{F}_{\tilde{V}}(w) & = \left\{ \begin{array}{ll}
    1, & w \in [0,\bar{F}_{\tilde{T}}(\tilde{a})) \\
    \frac{1}{\tilde{a}} \bar{F}^{-1}_{\tilde{T}}(w), & w \in [\bar{F}_{\tilde{T}}(\tilde{a}),1) \\
    0, & w > 1,
    \end{array} \right. 
\end{align*}
and 
\begin{align*}
    \bar{F}_{\tilde{B}\tilde{V}}(w) & = \left\{ \begin{array}{ll}
    \frac{a}{\tilde{a}}, & w \in [0,\bar{F}_{T}(a)) \\
    \frac{1}{\tilde{a}} \bar{F}^{-1}_{T}(w), & w \in [\bar{F}_{T}(a),1) \\
    0, & w > 1.
    \end{array} \right.
\end{align*} 
Since  $\bar{F}_{T} - \bar{F}_{\tilde{T}}$ has a single sign change with sign sequence is $+, -$, it follows that $\bar{F}_{\tilde{V}} - \bar{F}_{BV}$ also has a single sign change with sign sequence $+, -$. Applying  Lemma \ref{lem:phi} and \citet[Theorem 3.A.21]{SS:07} together shows $ X(N, \tilde{V}) \leq_{\rm cx}  X(N,B V)$. From Theorem \ref{thm:Rep}, $\tilde{M}(\tilde{a}) \stackrel{d}{=} \sum_{k=1}^{\Lambda(\tilde{a})} \tilde{X}_{k}$ and $ M(a)  \stackrel{d}{=} \sum_{k=1}^{\Lambda(a)} X_{k}$, where $\tilde{X}_{k} \stackrel{d}{=} X(N, \tilde{V})$ and $X_{k}  \stackrel{d}{=} X(N,V)$. Let $B_{1}, B_{2},\ldots$ be a sequence of independent $\mathsf{Ber}(a/\tilde{a})$ random variables that are  also independent of $X_{1},X_{2},\ldots$ By construction $BX(N,V) = X(N,BV)$. From the thinning property of the Poisson process, $ M(a) \stackrel{d}{=} \sum_{k=1}^{\Lambda(\tilde{a})} B_{k} X_{k}$. As the convex order is closed under random sums, we see $ \tilde{M}(\tilde{a}) \leq_{\rm cx} M(a)$. Letting $ a\to \infty$ and noting that the convex order is closed under weak limits, we see $\tilde{M}(\infty) \leq_{\rm cx} M(\infty)$.
\end{proof}

\subsection{Asymptotic normality}
As noted previously, for the host's parasite burden to converge to a normal distribution, the coefficient of variation must tend to 0. In the Tallis-Leyton model, this is only possible when the rate of infective contacts tends to infinity. 

\begin{theorem}
    Suppose there exists positive constants $\epsilon,\ \delta$ and $C$ such that 
    \[
    \left| G_{N}(1 + \omega) - \left(1 + G_{N}'(1) \omega + \frac{1}{2} G_{N}''(1) \omega^{2} \right) \right| \leq C \left| \omega \right|^{2+\epsilon}
    \]
    for all $\omega \in \mathbb{C}$ such that $|\omega| < \delta$. Then
    \[
    \lim_{\lambda \to \infty} \frac{M(a) - \mu(a)}{\sigma(a)} \stackrel{d}{=} \mathsf{N}(0,1).
    \]
    
\end{theorem}

\begin{proof}
The characteristic function of $M(a)$ is $G_{M}(e^{i\omega};a)$. We aim to show that 
\begin{equation}
\lim_{\lambda \to \infty} e^{-i\tfrac{\omega \mu(a)}{\sigma(a)}} G_{M}(e^{i \tfrac{\omega}{\sigma(a)}};a) = \exp\left(- \tfrac{1}{2} \omega^{2}\right). \label{eq:ANlimit}   
\end{equation}
The result then follows by L\'{e}vy's convergence theorem. Define 
\[
R_{N}(\omega) =  G_{N}(1 + \omega) - \left(1 + G_{N}'(1) \omega + \frac{1}{2} G_{N}''(1) w^{2} \right).
\]
For non-negative integers $n$ and real $x$ define 
\[
R_{n}(x) = e^{ix}  - \sum_{k=0}^{n} \frac{(ix)^{k}}{k!}.
\]
Then $ |R_{0}(x)| \leq \min(2, |x|)$ and 
\begin{equation}
\left| R_{n}(x) \right| \leq \min \left(\frac{2|x|^{n}}{n!}, \frac{|x|^{n+1}}{(n+1)!} \right).    \label{eq:ANbound}
\end{equation}
\citep[pg 183]{Williams:91}. Note that
\begin{equation*}
e^{-i\tfrac{\omega \mu(a)}{\sigma(a)}} G_{M}(e^{i \tfrac{\omega}{\sigma(a)}};a)  = \exp \left(\lambda \int^a_0 \left[ G_{N}(1 + \bar{F}_{T}(a-s) (e^{i \tfrac{\omega}{\sigma(a)}}-1) ) - 1 \right] ds  - i \tfrac{\omega \mu(a)}{\sigma(a)} \right).   
\end{equation*}
From the expressions for $R_N $ and $\mu(a)$,
%\begin{multline*}
%    \lambda \int^a_0 \left[ G_{N}(1 + \bar{F}_{T}(a-s) (e^{i \tfrac{\omega}{\sigma(a)}}-1) ) - 1 \right] ds  - i \tfrac{\omega \mu(a)}{\sigma(a)}  \\
%      =  \lambda G_{N}'(1) \left(\int^a_0  \bar{F}_{T}(a-s) ds\right) \left(e^{i\frac{\omega}{\sigma(a)}}  - i \frac{\omega}{\sigma(a)} -1 \right)  \\ + \frac{\lambda}{2}  G_{N}''(1) \left(\int^a_0\bar{F}_{T}^2(a-s) ds\right)  \left(e^{i\tfrac{\omega}{\sigma(a)}} - 1\right)^2  
%       + \lambda \int_0^a R_{N}\left(\bar{F}_{T}(a-s) (e^{i\tfrac{\omega}{\sigma(a)}} -1 ) \right)  ds .
%\end{multline*}
\begin{multline*}
     \lambda \int^a_0 \left[ G_{N}(1 + \bar{F}_{T}(a-s) (e^{i \tfrac{\omega}{\sigma(a)}}-1) ) - 1 \right] ds  - i \tfrac{\omega \mu(a)}{\sigma(a)} \\
      =  \lambda G_{N}'(1) \left(\int^a_0  \bar{F}_{T}(a-s) ds\right) R_{1} \left(\frac{\omega}{\sigma(a)} \right)    + \frac{\lambda}{2}  G_{N}''(1) \left(\int^a_0\bar{F}_{T}^2(a-s) ds\right)   R_{0} \left(\frac{\omega}{\sigma(a)} \right)^2  \\
       + \lambda \int_0^a R_{N}\left(\bar{F}_{T}(a-s) (e^{i\tfrac{\omega}{\sigma(a)}} -1 ) \right)  ds 
\end{multline*}
From the expression for $\sigma^{2}(a)$ and the fact that  
\begin{align*}
   R_{1}(x)^2  + x^{2} = R_2(2x) - 2R_2(x),
\end{align*}
we obtain
\begin{multline*}
    \lambda \int^a_0 \left[ G_{N}(1 + \bar{F}_{T}(a-s) (e^{i \tfrac{\omega}{\sigma(a)}}-1) ) - 1 \right] ds  - i \tfrac{\omega \mu(a)}{\sigma(a)} \\
    = -\frac{\omega^{2}}{2} +
      \lambda G_{N}'(1) \left(\int^a_0  \bar{F}_{T}(a-s) ds\right) R_{2} \left(\frac{\omega}{\sigma(a)} \right)   \\
       + \frac{\lambda}{2}  G_{N}''(1) \left(\int^a_0\bar{F}_{T}^2(a-s) ds\right)   \left( R_{2} \left(\frac{2\omega}{\sigma(a)} \right) - 2 R_{2} \left(\frac{\omega}{\sigma(a)} \right) \right) \\
       + \lambda \int_0^a R_{N}\left(\bar{F}_{T}(a-s) (e^{i\tfrac{\omega}{\sigma(a)}} -1 ) \right)  ds  
\end{multline*}
Using the bound (\ref{eq:ANbound}) and the fact that $\sigma^{2}(a) \propto \lambda $, we see
\[
\lim_{\lambda\to\infty} \lambda G_{N}'(1) \left(\int^a_0  \bar{F}_{T}(a-s) ds\right) R_{2} \left(\frac{\omega}{\sigma(a)} \right) = 0 
\]
and
\[
\lim_{\lambda\to\infty} \frac{\lambda}{2}  G_{N}''(1) \left(\int^a_0\bar{F}_{T}^2(a-s) ds\right)   \left( R_{2} \left(\frac{2\omega}{\sigma(a)} \right) - 2 R_{2} \left(\frac{\omega}{\sigma(a)} \right) \right) = 0.
\]
Finally, using $ \left| R_{N}(\omega) \right| \leq C|\omega|^{2+\epsilon}$ together with the bound (\ref{eq:ANbound}) and the fact that $\sigma^{2}(a) \propto \lambda $, we see
\[
\lim_{\lambda\to\infty} \lambda \int_0^a R_{N}\left(\bar{F}_{T}(a-s) (e^{i\tfrac{\omega}{\sigma(a)}} -1 ) \right)  ds  = 0.
\]
Hence, the limit (\ref{eq:ANlimit}) holds.
\end{proof}

Figure \ref{Fig:Normal} compares the probability mass function from the Tallis-Leyton model with the probability density function of the approximating normal distribution. The host was aged 3 with $ N \sim \mathsf{NB}(1,1)$ and $T \sim \mathsf{Exp}(1)$. When $\lambda = 8$, the probability mass function still shows some right skewness. The normal approximation in this instance places a non-negligible probability on values less than zero. When $\lambda = 128$, the probability mass function is very close to symmetric and the normal distribution provides a good approximation. Figure \ref{Fig:CV} shows the Hoover and Gini indices together with the approximations based on asymptotic normality (\ref{eq:Hoover_approx}) and (\ref{eq:Gini_approx}). In this instance the approximations appear reasonably accurate even for $\lambda $ as small as 2, where the normal approximation fails.

\begin{figure}[h] 
    \centering
    \includegraphics{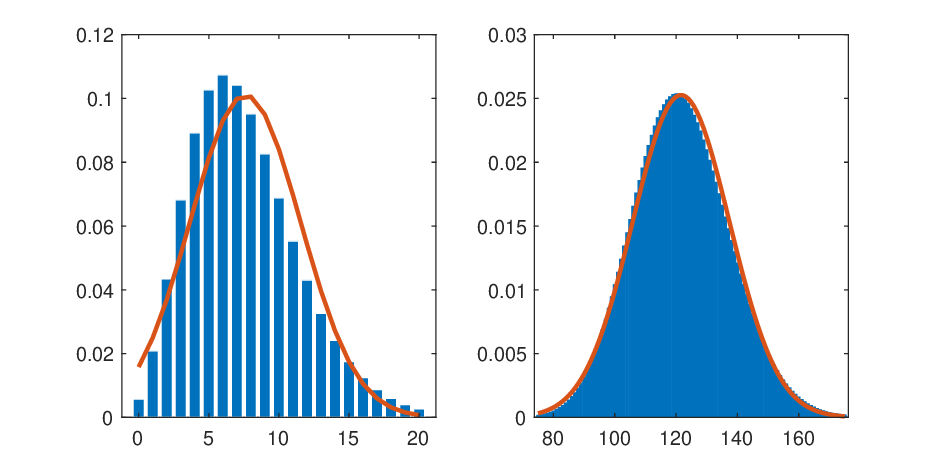}
    \caption{Probability mass function (blue bars) and approximating normal probability density function (red line) for a host aged 3 in the Tallis-Leyton model with $ N \sim \mathsf{NB}(1,1)$, $ T \sim \mathsf{Exp}(1)$, and $\lambda = 8$ (left) and $\lambda = 128$ (right).}
    \label{Fig:Normal}
\end{figure}

\begin{figure}[h] 
    \centering
    \includegraphics{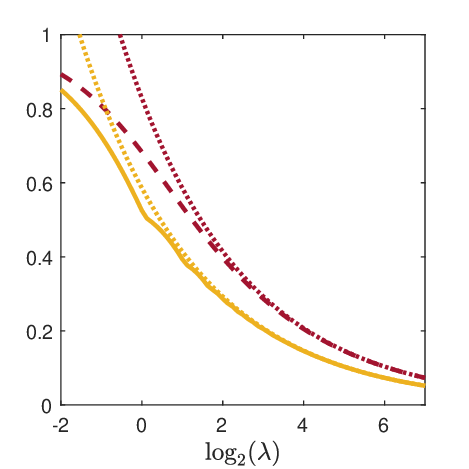}
    \caption{Hoover index (yellow line) and Gini index (purple dashed line) together with the asymptotic normal approximations (dotted lines) for a host aged 3 in the Tallis-Leyton model with $ N \sim \mathsf{NB}(1,1)$ and $ T \sim \mathsf{Exp}(1)$.}
    \label{Fig:CV}
\end{figure}

\section{Discussion}

In this paper, we examined the aggregation of a host's parasite burden in the Tallis-Leyton model, interpreting aggregation in terms of the Lorenz order. Our analysis showed that increasing the host's age or the rate of infective contacts decreases aggregation. Similarly, increasing the aggregation of the clumps of parasites that enter the host as infective contacts also increases the aggregation of the host's parasite burden. On the other hand, less variability in the parasite's lifetime distribution, as defined in the convex order, results in greater aggregation of the host's parasite burden.

Unfortunately, the population dynamics of parasites are often more complicated than what is represented in the Tallis-Leyton model. Some parasites need multiple hosts to complete its life cycle. Once a parasite finds a host it may be subject to intraspecific and interspecific competition for resources. Furthermore, parasites often interact with the host either by stimulating an immune response from the host or by increasing the host's mortality rate. 

\citet{Isham:95} proposed a simple stochastic model that incorporate parasite induced host mortality. In Isham's model, the host acquires parasites following the same dynamics as the Tallis-Leyton model and parasite lifetimes are assumed exponentially distributed. The important difference in Isham's model is that each parasite present in the host increases the host's death rate by a fixed amount $\alpha$. A complete analysis of Isham's model in terms of the Lorenz order is beyond the scope of this paper. In a special case, however, we can see that parasite induced host mortality increases aggregation of the parasite distribution, as interpreted in the Lorenz order. When the number of parasites that enter the host at an infective contact follows a geometric distribution, an explicit expression for the limiting distribution is possible. Specifically, if $ N \sim \mathsf{NB}(m,1)$, then 
\[
    M(\infty) \sim \mathsf{NB}\left(\frac{\lambda m}{ \mathbb{E} T + \alpha + \alpha m} , \frac{\lambda}{ \mathbb{E} T + \alpha + \alpha m}\right).
\]
As the negative binomial distribution is decreasing in Lorenz order in both mean and $k$, it follows that indices respecting the Lorenz order are increasing in the parasite induced host mortality rate. In contrast, the variance-to-mean ratio is $1  + m $ so it is not affected by the parasite induced mortality.

A complete examination Isham's model in terms of the Lorenz order may prove challenging. Even computing the Gini and Hoover indices may present difficulties since they require absolute moments, which are often not easily evaluated. In that case, the coefficient of variation may prove useful since it respects the Lorenz order, is easily evaluated, and can be used to approximate the Gini and Hoover indices when the distribution is approximately normal. 

\bigskip 

\noindent {\bf Data availability:} Data sharing not applicable to this article as no datasets were generated or analysed during the current study.

%%%%%%%%%%%%Reference list%%%%%%%%%%%%%%
%
% References should be in the following form (or the BibTeX file
% apt.bst should be used):
%
% For a journal:
% Surname, Initial (year). Title of paper. {\em Journal title}
% {\bf Vol,} page--range.
%
% For a book:
% Surname, Initial (year). {\em Book title}. Publisher, Address.
%
% Note the following example of a reference list.

\bibliographystyle{plainnat}
\bibliography{parasites}

\begin{thebibliography}{31}
\providecommand{\natexlab}[1]{#1}
\providecommand{\url}[1]{\texttt{#1}}
\expandafter\ifx\csname urlstyle\endcsname\relax
  \providecommand{\doi}[1]{doi: #1}\else
  \providecommand{\doi}{doi: \begingroup \urlstyle{rm}\Url}\fi

\bibitem[Abate and Whitt(1992)]{AB:92}
J.~Abate and W.~Whitt.
\newblock Numerical inversion of probability generating functions.
\newblock \emph{Oper. Res. Lett.}, 12:\penalty0 245--251, 1992.

\bibitem[Adell and Cal(1994)]{AdlC:94}
J.A. Adell and J.~De~La Cal.
\newblock Approximating gamma distributions by normalized negative binomial
  distributions.
\newblock \emph{J. Appl. Probab.}, 31:\penalty0 391--400, 1994.

\bibitem[Anderson and May(1978{\natexlab{a}})]{AM:78a}
R.M. Anderson and R.M. May.
\newblock Regulation and stability of host-parasite population interactions:
  {I}. regulatory processes.
\newblock \emph{J. Anim. Ecol.}, 47:\penalty0 219--247, 1978{\natexlab{a}}.

\bibitem[Anderson and May(1978{\natexlab{b}})]{AM:78b}
R.M. Anderson and R.M. May.
\newblock Regulation and stability of host-parasite population interactions:
  {II}. destabilizing processes.
\newblock \emph{J. Anim. Ecol.}, 47:\penalty0 249--267, 1978{\natexlab{b}}.

\bibitem[Arnold and Sarabia(2018)]{Arnold:87}
B.C. Arnold and J.M. Sarabia.
\newblock \emph{Majorization and the Lorenz order with applications in applied
  mathematics and economics}.
\newblock Springer, New York, 2018.

\bibitem[Barbour and Pugliese(2000)]{BP:00}
A.D. Barbour and A.~Pugliese.
\newblock On the variance-to-mean ratio in models of parasite distributions.
\newblock \emph{Adv. Appl. Probab.}, 32:\penalty0 701--719, 2000.

\bibitem[Gastwirth(1971)]{Gastwirth:1971}
J.L. Gastwirth.
\newblock A general definition of the {L}orenz curve.
\newblock \emph{Econometrica}, 39:\penalty0 1037–1039, 1971.

\bibitem[Gini(1924)]{Gini:24}
C.~Gini.
\newblock Measurement of inequality of incomes.
\newblock \emph{Econ. J.}, 31:\penalty0 124--126, 1924.

\bibitem[Herbert and Isham(2000)]{HI:2000}
J.~Herbert and V.~Isham.
\newblock Stochastic host-parasite interaction models.
\newblock \emph{J. Math. Biol.}, 40:\penalty0 343--371, 2000.

\bibitem[Holman et~al.(1983)Holman, Chaudhry, and Kashyap]{HCK:83}
D.F. Holman, M.L. Chaudhry, and B.R.K. Kashyap.
\newblock On the service system ${M^{X}/G/\infty}$.
\newblock \emph{Eur. J. Oper. Res.}, 13:\penalty0 142--145, 1983.

\bibitem[Isham(1995)]{Isham:95}
V.~Isham.
\newblock Stochastic models of host-macroparasite interaction.
\newblock \emph{Ann. Appl. Probab.}, 5:\penalty0 720--740, 1995.

\bibitem[Lorenz(1905)]{Lorenz:1905}
M.O. Lorenz.
\newblock Methods of measuring the concentration of wealth.
\newblock \emph{Publication of the American Statistical Association},
  9:\penalty0 209--219, 1905.

\bibitem[McPherson et~al.(2012)McPherson, Norman, Hoyle, Bron, and
  Taylor]{McPherson:12}
N.J. McPherson, R.A Norman, A.S. Hoyle, J.E. Bron, and N.G.H. Taylor.
\newblock Stocking methods and parasite-induced reductions in capture:
  Modelling {{\em Argulus foliaceus}} in trout fisheries.
\newblock \emph{J. Theor. Biol.}, 312:\penalty0 22--33, 2012.

\bibitem[McVinish and Lester(2020)]{ML:2020}
R.~McVinish and R.J.G. Lester.
\newblock Measuring aggregation in parasite populations.
\newblock \emph{J. R. Soc. Interface}, 17:\penalty0 20190886, 2020.

\bibitem[McVinish and Lester(2024)]{ML:2024}
R.~McVinish and R.J.G. Lester.
\newblock A graphical exploration of the relationship between parasite
  aggregation indices.
\newblock \emph{arXiv:2409.03186}, 2024.

\bibitem[Morrill et~al.(2023)Morrill, Poulin, and Forbes]{MPF:2023}
A.~Morrill, R.~Poulin, and M.R. Forbes.
\newblock Interrelationships and properties of parasite aggregation measures: a
  user’s guide.
\newblock \emph{Int. J. Parasitol.}, 53:\penalty0 763--776, 2023.

\bibitem[Peacock et~al.(2018)Peacock, Bouhours, Lewis, and
  Moln\'{a}r]{Peacock:18}
S.J. Peacock, J.~Bouhours, M.A. Lewis, and P.K. Moln\'{a}r.
\newblock Macroparasite dynamics of migratory host populations.
\newblock \emph{Theor. Popul. Biol.}, 120:\penalty0 29--41, 2018.

\bibitem[Pielou(1977)]{Pielou:77}
E.C. Pielou.
\newblock \emph{Mathematical ecology}.
\newblock Wiley, New York, 1977.

\bibitem[Poulin(1993)]{Poulin:1993}
R.~Poulin.
\newblock The disparity between observed and uniform distributions: a new look
  at parasite aggregation.
\newblock \emph{Int. J. Parasitol.}, 23:\penalty0 931--944, 1993.

\bibitem[Poulin(2007)]{Poulin:07}
R.~Poulin.
\newblock Are there general laws in parasite ecology?
\newblock \emph{Parasitol.}, 134:\penalty0 763--776, 2007.

\bibitem[Rosa and Pugliese(2002)]{RP:02}
R.~Rosa and A.~Pugliese.
\newblock Aggregation, stability, and oscillations in different models for
  host-macroparasite interactions.
\newblock \emph{Theor. Popul. Biol.}, 61\penalty0 (3):\penalty0 319--334, 2002.

\bibitem[Schreiber(2006)]{Schreiber:06}
S.J. Schreiber.
\newblock Host-parasitoid dynamics of a generalized thompson model.
\newblock \emph{J. Math. Biol.}, 52:\penalty0 719--732, 2006.

\bibitem[Schweder(1982)]{Schweder:1982}
T.~Schweder.
\newblock On the dispersion of mixtures.
\newblock \emph{Scand. J. Stat.}, 9:\penalty0 165--169, 1982.

\bibitem[Shaked and Shanthikumar(2007)]{SS:07}
M.~Shaked and J.~G. Shanthikumar.
\newblock \emph{Stochastic orders}.
\newblock Springer, New York, 2007.

\bibitem[Shaw and Dobson(1995)]{SD:95}
D.J. Shaw and A.P. Dobson.
\newblock Patterns of macroparasite abundance and aggregation in wildlife
  populations: a quantitative review.
\newblock \emph{Parasitol.}, 111:\penalty0 S111--S133, 1995.

\bibitem[Taguchi(1968)]{Taguchi:68}
T.~Taguchi.
\newblock Concentration-curve methods and structures of skew populations.
\newblock \emph{Ann. Inst. Stat. Math.}, 20:\penalty0 107--141, 1968.

\bibitem[Tallis and Leyton(1969)]{TL:1969}
G.M. Tallis and M.K. Leyton.
\newblock Stochastic models of populations of helminthic parasites in the
  definitive host. {I}.
\newblock \emph{Math. Biosci.}, 4:\penalty0 39--48, 1969.

\bibitem[{The MathWorks Inc.}(2022{\natexlab{a}})]{MATLAB}
{The MathWorks Inc.}
\newblock Matlab version: 9.13.0.2080170 (r2022b) update 1, 2022{\natexlab{a}}.

\bibitem[{The MathWorks Inc.}(2022{\natexlab{b}})]{vpa}
{The MathWorks Inc.}
\newblock Symbolic math toolbox (r2022b), 2022{\natexlab{b}}.

\bibitem[Williams(1991)]{Williams:91}
D.~Williams.
\newblock \emph{Probability with Martingales}.
\newblock Cambridge, 1991.

\bibitem[Wilson et~al.(2001)Wilson, rnstad, Dobson, Merler, Poglayen, Randolph,
  Read, and Skorping]{Wilson:2001}
K.~Wilson, O.~Bj\o rnstad, A.~Dobson, S.~Merler, G.~Poglayen, S.~Randolph,
  A.~Read, and A.~Skorping.
\newblock Heterogeneities in macroparasite infections: Patterns and processes.
\newblock In P.~Hudson, A.~Rizzoli, B.~Grenfell, H.~Heesterbeek, and A.~Dobson,
  editors, \emph{The Ecology of Wildlife Diseases}, pages 6--44. Oxford
  University Press, 2001.

\end{thebibliography}

\end{document}